\documentclass[10pt,conference]{IEEEtran}
\IEEEoverridecommandlockouts
\usepackage{cite}
\usepackage{amsmath,amssymb,amsfonts}
\usepackage{algorithmic}
\usepackage{graphicx}
\usepackage{textcomp}
\usepackage{xcolor}
\usepackage{amsthm}
\usepackage{dsfont}
\usepackage[ruled,linesnumbered,lined]{algorithm2e}
\usepackage{algorithmic}
\usepackage{float}
\usepackage[T1]{fontenc}
\usepackage{lmodern}

\usepackage{times}   % or \usepackage{mathptmx}
\usepackage{newtxtext,newtxmath} % more modern

\newtheorem{assumption}{Assumption}
\newtheorem{lemma}{Lemma}
\def\BibTeX{{\rm B\kern-.05em{\sc i\kern-.025em b}\kern-.08em
    T\kern-.1667em\lower.7ex\hbox{E}\kern-.125emX}}
\begin{document}

\title{Bayesian Optimization for Non-Cooperative Game-Based Radio Resource Management
}

\author{\IEEEauthorblockN{Yunchuan Zhang\IEEEauthorrefmark{1},
Jiechen Chen\IEEEauthorrefmark{2},
Junshuo Liu\IEEEauthorrefmark{3} and
Robert C. Qiu\IEEEauthorrefmark{3}}
\IEEEauthorblockA{\IEEEauthorrefmark{1}School of Information Engineering,
Wuhan University of Technology,
Wuhan, China}
\IEEEauthorblockA{\IEEEauthorrefmark{2}Department of Engineering,
King's College London,
London, UK}
\IEEEauthorblockA{\IEEEauthorrefmark{3}School of Electronic Information and Communications, Huazhong University of Science and Technology, Wuhan, China\\
Email: \{yunchuan.zhang, jiechen.chen\}@kcl.ac.uk, \{junshuo\_liu, caiming\}@hust.edu.cn} 
}

\maketitle

\begin{abstract}
Radio resource management in modern cellular networks often calls for the optimization of complex utility functions that are potentially conflicting between different base stations (BSs). Coordinating the resource allocation strategies efficiently across BSs to ensure stable network service poses significant challenges, especially when each utility is accessible only via costly, black-box evaluations. This paper considers formulating the resource allocation among spectrum sharing BSs as a non-cooperative game, with the goal of aligning their allocation incentives toward a stable outcome. To address this challenge, we propose PPR-UCB, a novel Bayesian optimization (BO) strategy that learns from sequential decision-evaluation pairs to approximate pure Nash equilibrium (PNE) solutions. PPR-UCB applies martingale techniques to Gaussian process (GP) surrogates and constructs high probability confidence bounds for utilities uncertainty quantification. Experiments on downlink transmission
power allocation in a multi-cell multi-antenna system demonstrate the efficiency of PPR-UCB in identifying effective equilibrium solutions within a few data samples.
\end{abstract}


\section{Introduction}
Radio resource management of modern cellular communication system requires the optimization of complex utilities which may involve potential \emph{conflicts} between different segments of the network. To ensure stable network operations during optimizations, one promising way is via game-theoretic methods \cite{chen2018resource}. For instance, heterogeneous spectrum providers configure the exclusion zones in dynamic spectrum access (DSA) networks via analyzing \emph{Nash equilibrium} (NE) in a non-cooperative game \cite{salama2020privacy}, but conventional mathematical programming methods for NE often require analytical expressions of the utilities which may not be available in practice. When the utilities are in the form of \emph{black-box} functions, the evaluation of a NE must rely on querying the utilities of each involved entity.

For concreteness, this paper investigates the downlink transmission power control game with \emph{unknown} utilities in a multi-cell multi-antenna communication system \cite{zhang2021non}. As illustrated in Fig. \ref{fig: intro flow}, a central optimizer assigns a set of transmission power actions to different \emph{players}, representing spectrum sharing BSs managed by mobile network operators (MNOs). The goal is to design an optimization policy that ensures no players have incentives to deviate from the allocated power actions. Thus, the suggested action set aligns closely with the individual rationality of the players. This equilibrium solution ensures a stable operation of the network as the players strictly adhere to the allocated actions of the central optimizer.

To tackle the non-cooperative games with black-box utilities, prior works applied reinforcement learning (RL) for identifying NE solutions \cite{bai2020near,li2025deep}, and provided practical use cases in communications systems \cite{shi2021make,nagib2023safe}. In particular, reference \cite{naparstek2018deep} introduced a fully distributed deep RL architecture for DSA that can be flexibly adapted to general complex real-world settings. However, RL algorithms typically rely on large amounts of utilities observations which may incur undesirable evaluation costs on computational overhead or communications latency.

\emph{Bayesian optimization} (BO) is a common sample-efficient framework for \emph{costly-to-evaluate} black-box optimization problems \cite{frazier2018tutorial}, and its variants have been widely applied to applications in wireless systems \cite{zhang2023bayesian,maggi2021bayesian}. The first attempt to seek for NE via BO can be found in \cite{picheny2019bayesian}; while reference \cite{tay2023no} provided a formal NE regret guarantee for using BO with confidence sets. However, these studies require making strong assumptions on the function space knowledge regarding the utilities and evaluating the NP-hard maximal information gain metrics which are impractical in real-world scenarios.

In this paper, we introduce a novel BO policy, namely \emph{prior-posterior ratio upper confidence bounds} (PPR-UCB) that aims at approximating a \emph{pure Nash equilibrium} (PNE) solution in a non-cooperative transmission power control game. The main contributions are as follows.
\begin{itemize}
    \item We introduce PPR-UCB, a novel BO policy tailored for evaluation of approximate PNEs in general-sum games with black-box utility functions. PPR-UCB adopts martingale techniques to efficiently construct confidence sets for the probabilistic utilities and deviation incentives of each involved players.
    \item We provide a coverage analysis of the \emph{any-time valid} confidence sets constructed by the prior-posterior ratio martingales. The theoretical results demonstrates the reliability of the proposed PPR-UCB policy.
    \item We validate the performance of the proposed PPR-UCB in a multi-cell cellular network. Empirical results for the non-cooperative power control game provide insights into the potential benefits of using BO for conflicts management in communications systems.
\end{itemize}

\section{Problem Formulation}\label{sec: pf}
To exemplify the application of the proposed sequential black-box optimization policy, we consider the non-cooperative game-based downlink transmission power control problem studied in \cite{zhang2021non}. The targeting cellular communication network consists of $N$ cells with one BS, each belonging to a private MNO that provides service for $M$ user equipments (UEs). The MNOs in the network operates on a shared frequency band.

\begin{figure}[t]

%\begin{minipage}[b]{1.0\linewidth}
  \centering
  \centerline{\includegraphics[scale=0.43]{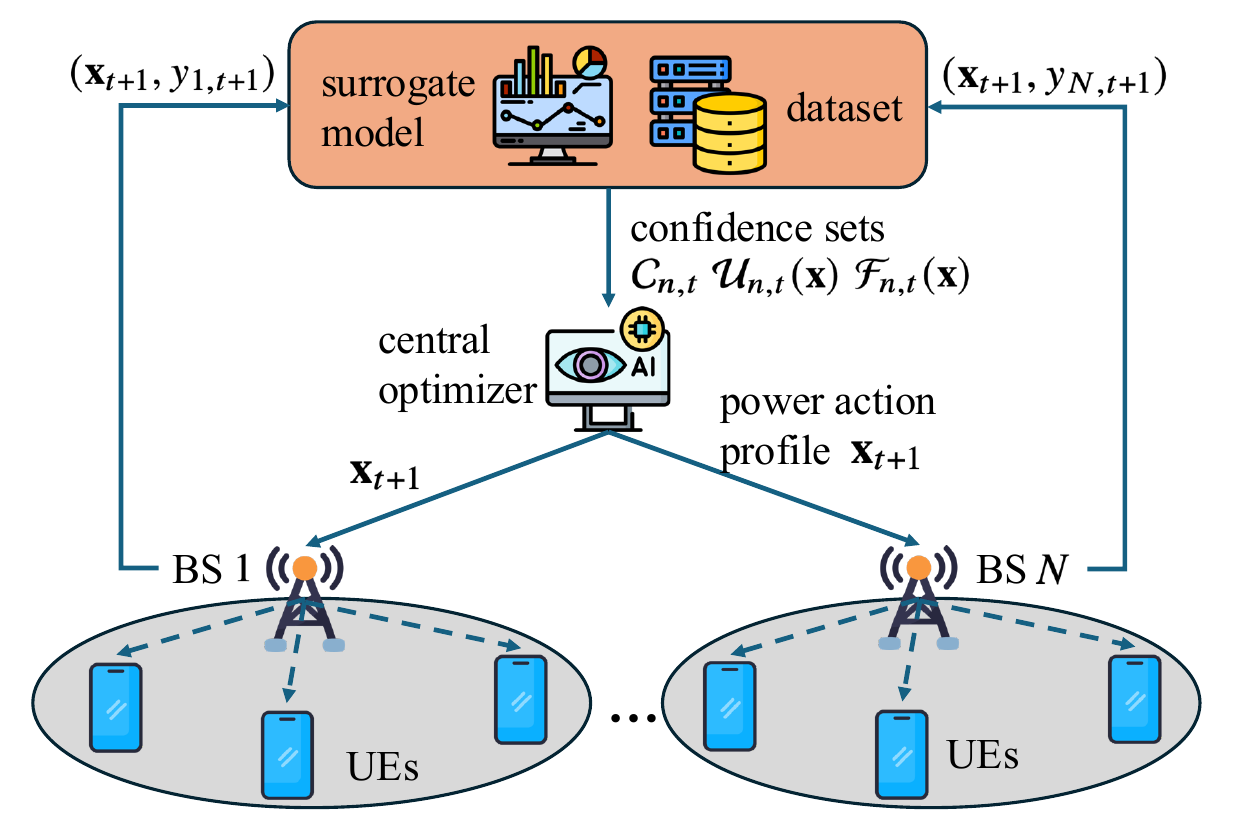}}
  %\vspace{-0.2cm}
  \caption{This paper studies a setting in which a central optimizer uses BO to approximate the pure Nash equilibrium (PNE) for a non-cooperative downlink transmission power control game with costly-to-evaluate black-box utility functions of $N$ BSs. At any time $t+1$, the central optimizer assigns an action profile $\mathbf{x}_{t+1}$ to all BSs. As a result, the optimizer receives noisy utility feedback $y_{n,t}$ about the corresponding utility value $u_n(\mathbf{x}_{t+1})$ for all BSs $n\in\mathcal{N}$. The goal is to approach a solution in the $\epsilon$-PNE set \eqref{eq: epsilon pne}, where $\epsilon\geq 0$ represents the dissatisfaction tolerance.}
  \label{fig: intro flow}
\vspace{-0.4cm}
\end{figure}

As shown in Fig. \ref{fig: intro flow}, the central coordinator is tasked with selection of an action profile $\mathbf{x}=[\mathbf{x}_1,...,\mathbf{x}_N]\in\mathcal{X}$ for $N$ BSs, where $\mathbf{x}_n=[x_{n,1},...,x_{n,M}]\subset\mathds{R}^{M}$ represents the selected power action for BS $n\in\mathcal{N}=[1,...,N]$ with each entry $x_{n,m}$ describing the downlink transmission power to its serving UE $m\in\mathcal{M}=[1,...,M]$. The propagation channel matrix between BS $n$ and any UE $m$ associated with BS $n'$ is modeled as
\begin{align}
    \mathbf{H}_{n,n'(m)}=10^{\frac{-\mathrm{PL}(d_{n,n'(m)})}{20}}\beta_{n,n'(m)}\mathbf{G}_{n,n'(m)}, \label{eq: channel matrix}
\end{align}
where the pathloss $\mathrm{PL}(d_{n,n'(m)})$ in dB is a function of the distance from BS $n$ to UE $m$ served by BS $n'$; the slow fading factor $\beta_{n,n'(m)}$ depends on whether the corresponding transmission is in line-of-sight (LOS) or non-line-of-sight (NLOS) as specified in 3GPP TR 38.901 \cite{3gpp38901}; and the fast fading $\mathbf{G}_{n,n'(m)}$ in the form of $N_R\times N_T$ matrix has i.i.d. complex Gaussian entries with $N_T$ being the number of transmit antennas at each BS and $N_R$ being the number of receiving antennas at each UE. 

Each MNO is \emph{self-interested}, and attempts to maximize the utility of its own BS. Specifically, the utility function for BS $n\in\mathcal{N}$ at the assigned action profile $\mathbf{x}$ is described by the discounted sum spectral efficiency across the associated UEs, which is defined as
\begin{align}
    u_n(\mathbf{x})=&-\lambda\bigg(\sum_{m=1}^Mx_{n,m}\bigg)+\Bigg[\sum_{m=1}^M\log\big|\mathbf{I}\nonumber\\&+x_{n,m}\boldsymbol{\Gamma}_{n,n(m)}^{-1}(\mathbf{x})\mathbf{H}_{n,n(m)}\mathbf{H}_{n,n(m)}^{\sf H}\big|\Bigg],\label{eq: utility}
\end{align}
where $\lambda\geq0$ is a non-negative discount factor; $\mathbf{I}$ is the $N_R\times N_R$ identity matrix; and $\boldsymbol{\Gamma}_{n,n(m)}^{-1}(\mathbf{x})$ represents the $N_R\times N_R$ interference-plus-noise covariance matrix for the transmission from BS $n$ to its serving UE $m$, i.e,
\begin{align}
    \boldsymbol{\Gamma}_{n,n(m)}^{-1}(\mathbf{x})=&10^{\frac{\sigma_h^2}{10}}\mathbf{I}+\sum_{m'=1,m'\neq m}^{M}x_{n,m'}\mathbf{H}_{n,n(m)}\mathbf{H}^{\sf H}_{n,n(m)}\nonumber\\&+\sum^N_{n'=1,n'\neq n}\sum_{m=1}^Mx_{n',m}\mathbf{H}_{n',n'(m)}\mathbf{H}^{\sf H}_{n',n'(m)}
\end{align}
with $\sigma_h^2$ being the channel noise power in dB scale.

Upon receiving the action profile $\mathbf{x}$ from the central coordinator, each BS $n\in\mathcal{N}$ quantifies its \emph{dissatisfaction} on the assigned action $\mathbf{x}$ via the non-negative regret
\begin{align}
    f_n(\mathbf{x})=\max\limits_{\mathbf{x}_n'}u_n(\mathbf{x}_n',\mathbf{x}_{-n})-u_n(\mathbf{x})\geq 0,\label{eq: player n regret}
\end{align}
where $\mathbf{x}_{-n}$ represents the actions taken by all BSs except for BS $n$. The dissatisfaction \eqref{eq: player n regret} indicates that any BS $n$ may unilaterally deviate from the assigned action $\mathbf{x}_n$ if doing so can improve its utility. These deviation incentives across all BSs degrade the \emph{stability} of the network operations. To this end, the goal of the central optimizer is to seek for a set of $\epsilon$-PNE solutions $\mathbf{x}^{(\epsilon)}$ defined as
\begin{align}
    \mathcal{X}^{(\epsilon)}:=\{\mathbf{x}^{(\epsilon)}\in\mathcal{X}|f_n(\mathbf{x}^{(\epsilon)})\leq\epsilon\,\, \text{for}\,\, n\in \mathcal{N}\},\label{eq: epsilon pne}
\end{align}
for some tolerance parameter $\epsilon\geq 0$.

Evaluating the utility functions \eqref{eq: utility} and the regret \eqref{eq: player n regret} for all feasible action profiles to find the $\epsilon$-PNE set \eqref{eq: epsilon pne} is intractable as the scale of the network increases. Therefore, we adopt an online learning-based optimizer that sequentially selects an action profile $\mathbf{x}_t$ at time $t$, and receives the noisy observations from each BS $n\in\mathcal{N}$, i.e.,
\begin{align}
    y_{n,t}=u_n(\mathbf{x}_t)+z_{n,t},\label{eq: nosiy observation}
\end{align}
where the observation noise variables $z_{n,t}\sim\mathcal{N}(0,\sigma^2)$ are independent across all BSs. The decision-making process relies on the collections of previous observations
\begin{align}
    \mathcal{D}_{n,t-1}=\{(\mathbf{x}_1,y_{n,1}),...,(\mathbf{x}_{t-1},y_{n,t-1})\}\label{eq: dataset}
\end{align}
from all BSs.

\section{BO for Non-cooperative Games}\label{sec: bo part}
In this section, we introduce \emph{prior-posterior ratio upper confidence bounds} (PPR-UCB) by first presenting the confidence sets construction for utilities and regret, and then detailing the action profiles acquisition process. We start with a brief review of the surrogate model adopted in BO policy.

\subsection{Gaussian Process}\label{ssec: gp}
As anticipated, BO does not require a precise expression of the utility function for calculating the $\epsilon$-PNE set \eqref{eq: epsilon pne}, but works on black-box evaluations of the utility that is modeled by a \emph{probabilistic surrogate model}. Gaussian process (GP) is the typical model of choice as it provides analytical posterior inference with small datasets.

Specifically, for each BS $n$, a separate GP is deployed at the central optimizer to model the individual utility function $u_n(\mathbf{x})$ by assuming that, for any collection $\mathbf{X}_{t}=[\mathbf{x}_{1},...,\mathbf{x}_{t}]^{\sf T}$ of action profiles, the corresponding utility observations $\mathbf{y}_{n,t}=[y_{n,1},...,y_{n,t}]^{\sf T}$ follow a multivariate Gaussian distribution $\mathcal{N}(\mathbf{0},\mathbf{K}(\mathbf{X}_{t}))$, with $t\times 1$ zero mean vector $\mathbf{0}$, and $t\times t$ covariance matrix $\mathbf{K}(\mathbf{X}_{t})$ given by
\begin{align}
    \mathbf{K}(\mathbf{X}_{t})=\begin{bmatrix}
k(\mathbf{x}_{1},\mathbf{x}_{1}) & ... & k(\mathbf{x}_{1},\mathbf{x}_{t})\\
\vdots & \ddots & \vdots \\ k(\mathbf{x}_{t-1},\mathbf{x}_{1}) & ... & k(\mathbf{x}_{t},\mathbf{x}_{t})
\end{bmatrix}, \label{eq: covar mat}
\end{align}
where each entry is obtained via a positive semidefinite kernel function $k(\mathbf{x},\mathbf{x}')$. Intuitively, the role
of the kernel function is to measure the similarity between inputs $\mathbf{x}$ and $\mathbf{x}'$ in terms of the respective utility values. The kernel function may be chosen, for instance, as the radial basis function (RBF) kernel
\begin{align}
    k(\mathbf{x},\mathbf{x}')=\exp\bigg(-\frac{||\mathbf{x}-\mathbf{x}'||^2}{2l^2}\bigg),\label{eq: rbf kernel}
\end{align}
where the lengthscale parameter $l>0$ controls the smoothness of the kernel function.

Given the collected dataset $\mathcal{D}_{n,t}$ in \eqref{eq: dataset}, the posterior distribution of utility value $u_n(\mathbf{x})$ at any action profile $\mathbf{x}$ is a Gaussian distribution
\begin{align}
    p(u_n(\mathbf{x})|\mathcal{D}_{n,t})=\mathcal{N}(\mu_{n,t}(\mathbf{x}),\sigma^2_{t}(\mathbf{x})),\label{eq: gp posterior}
\end{align}
where 
\begin{subequations}
    \begin{equation}
        \mu_{n,t}(\mathbf{x})=\mathbf{k}_{t}(\mathbf{x})^{\sf T}(\mathbf{K}(\mathbf{X}_{t})+\sigma^2\mathbf{I})^{-1}\mathbf{y}_{n,t},\label{eq: gp posterior mean}
    \end{equation}
    \begin{equation}
        \sigma^2_{t}(\mathbf{x})=k(\mathbf{x},\mathbf{x})-\mathbf{k}_{t}(\mathbf{x})^{\sf T}(\mathbf{K}(\mathbf{X}_{t})+\sigma^2\mathbf{I})^{-1}\mathbf{k}_{t}(\mathbf{x}),\label{eq: gp posterior variance}
    \end{equation}
\end{subequations}
with the $t\times 1$ cross-variance vector $\mathbf{k}_{t}(\mathbf{x})=[k(\mathbf{x},\mathbf{x}_1),...,k(\mathbf{x},\mathbf{x}_{t})]^{\sf T}$. Note that, the GP posterior variance $\sigma^2_{t}(\mathbf{x})$ in \eqref{eq: gp posterior variance} is the same across all players since it only depends on the previously selected action profiles $\mathbf{X}_{t}$ and the current action profile $\mathbf{x}$ to be measured.

\subsection{Prior-Posterior Ratio Confidence Sets}\label{ssec: cs}
To quantify the residual uncertainty on the inference of the utility function $u_n(\mathbf{x})$ at any candidate action profile $\mathbf{x}$, PPR-UCB constructs the confidence sets for the model parameters with probabilistic coverage guarantee. To proceed, we make the following regularity assumption on the utility functions of each BS $n\in\mathcal{N}$.
\begin{assumption}[Probabilistic Utility Function]\label{assumption: bayesian linear regression}
    The utility function $u_n(\mathbf{x})$ at each BS $n$ can be decomposed as
    \begin{align}
        u_n(\mathbf{x})=\psi_n(\mathbf{x})^{\sf T}\boldsymbol{\theta}^*_n,\label{eq: linear function}
    \end{align}
    where $\psi_n(\mathbf{x})$ is a $D\times 1$ feature vector; and the unknown parameters vector $\boldsymbol{\theta}^*_n\in\mathds{R}^{D\times 1}$ is drawn from an isotropic Gaussian prior distribution $p(\boldsymbol{\theta})$, i.e., $\boldsymbol{\theta}^*_n\sim\mathcal{N}(\mathbf{0},\mathbf{I}_{D})$. The vectors $\{\boldsymbol{\theta}_n^*\}_{n\in\mathcal{N}}$ are mutually independent.
\end{assumption}

Assumption \ref{assumption: bayesian linear regression} can serve as an arbitrarily accurate approximation of the typical assumption that the utility functions are drawn in an i.i.d. manner from the GP model in \eqref{eq: gp posterior} with RBF kernel \cite{dai2020federated}. In fact, the RBF kernel can be approximated arbitrarily well by the inner product
\begin{align}
    k(\mathbf{x},\mathbf{x}')\approx\psi_n(\mathbf{x})^{\sf T}\psi_n(\mathbf{x}')\label{eq: rff decomposition}
\end{align}
with a sufficiently high-dimensional feature vector $\psi_n(\mathbf{x})\in\mathds{R}^{D\times 1}$. In practice, consider the random Fourier features (RFFs)
\begin{align}
    \psi_n(\mathbf{x})=[\sqrt{2/D}\cos{(\mathbf{s}_{n,i}^{\sf T}\mathbf{x}+b_{n,i})}]_{i=1}^{D},\label{eq: Bayesian rff expression}
\end{align}
with $d\times 1$ vectors $\mathbf{s}_{n,i}$ drawn i.i.d. from $\mathcal{N}(0,2l\mathbf{I}_d)$ and scalars $b_{n,i}$ sampled i.i.d. from the uniform distribution over the interval $[0,2\pi]$. With an increasing dimension $D\rightarrow\infty$ of the RFFs, the approximation error in \eqref{eq: rff decomposition} vanishes with high probability \cite{dai2020federated}.

At each time $t$, PPR-UCB evaluates the posterior distribution of parameters vector $\boldsymbol{\theta}_n$ for BS $n\in\mathcal{N}$ under Assumption \ref{assumption: bayesian linear regression} as follows
\begin{align}
    p(\boldsymbol{\theta}_n|\mathcal{D}_{n,t})\propto p(\boldsymbol{\theta}_n)p(\mathcal{D}_{n,t}|\boldsymbol{\theta}_n)=\mathcal{N}(\boldsymbol{\mu}_{n,t},\sigma^2\boldsymbol{\Sigma}_{n,t}^{-1}),\label{eq: posterior over theta}
\end{align}
where
\begin{subequations}
    \begin{equation}
        \boldsymbol{\mu}_{n,t}=\boldsymbol{\Sigma}_{n,t}^{-1}\boldsymbol{\Psi}_{n,t}^{\sf T}\mathbf{y}_{n,t},\label{eq: theta posterior mean}
    \end{equation}
    \begin{equation}
        \boldsymbol{\Sigma}_{n,t}=\boldsymbol{\Psi}_{n,t}^{\sf T}\boldsymbol{\Psi}_{n,t}+\sigma^2\mathbf{I}_{D}\label{eq: theta posterior variance}
    \end{equation}
\end{subequations}
with $\boldsymbol{\Psi}_{n,t}=[\psi_n(\mathbf{x}_1),...,\psi_n(\mathbf{x}_{t})]^{\sf T}$ being the $t\times D$ matrix of the feature vectors. Therefore, the mean and variance of the GP posterior distribution \eqref{eq: gp posterior} for the utility value $u_n(\mathbf{x})$ can be obtained as
\begin{subequations}
    \begin{equation}
        \mu_{n,t}(\mathbf{x})=\psi_n(\mathbf{x})^{\sf T}\boldsymbol{\Sigma}_{n,t}^{-1}\boldsymbol{\Psi}_{n,t}^{\sf T}\mathbf{y}_{n,t},\label{eq: utility gp posterior mean}
    \end{equation}
    \begin{align}
        \sigma^2_t(\mathbf{x})=&\psi_n(\mathbf{x})^{\sf T}\psi_n(\mathbf{x})\nonumber\\&-\psi_n(\mathbf{x})^{\sf T}\boldsymbol{\Psi}_{n,t}^{\sf T}(\boldsymbol{\Psi}_{n,t}^{\sf T}\boldsymbol{\Psi}_{n,t}+\sigma^2\mathbf{I}_{D})^{-1}\boldsymbol{\Psi}_{n,t}\psi_n(\mathbf{x}).\label{eq: utility gp posterior variance}
    \end{align}
\end{subequations}

Using the posterior \eqref{eq: posterior over theta} of the parameters vector $\boldsymbol{\theta}_n$ for BS $n$ at each time $t$, PPR-UCB evaluates the prior-posterior ratio
\begin{align}
    \ell_{n,t}(\boldsymbol{\theta}_n)&=\frac{p(\boldsymbol{\theta}_n)}{p(\boldsymbol{\theta}_n|\mathcal{D}_{n,t})}\nonumber\\&=\frac{\sigma^D}{\sqrt{\text{det}(\boldsymbol{\Sigma}_{n,t})}}\exp\bigg(-\frac{1}{2}||\boldsymbol{\theta}_n||^2\nonumber\\&\hspace{0.5cm}+\frac{1}{2\sigma^2}(\boldsymbol{\theta}_n-\boldsymbol{\mu}_{n,t})^{\sf T}\boldsymbol{\Sigma}_{n,t}(\boldsymbol{\theta}_n-\boldsymbol{\mu}_{n,t})\bigg),\label{eq: pp ratio}
\end{align}
where $\text{det}(\boldsymbol{\Sigma}_{n,t})$ is the determinant of covariance matrix $\boldsymbol{\Sigma}_{n,t}$. Smaller values of the prior-posterior ratio indicates that the parameters vector $\boldsymbol{\theta}_n$ has a large posterior probability density, suggesting that $\boldsymbol{\theta}_n$ aligns better with the evidence provided by the collected data under Assumption \ref{assumption: bayesian linear regression}.

At any time $t\geq 1$, using the prior-posterior ratio in \eqref{eq: pp ratio}, PPR-UCB builds confidence sets for parameters vector $\boldsymbol{\theta}_n$ as
\begin{align}
    \mathcal{C}_{n,t}&=\{\boldsymbol{\theta}_n|\ell_{n,t}(\boldsymbol{\theta}_n)\leq1/\delta\}\nonumber\\&=\big\{\boldsymbol{\theta}_n\big|(\boldsymbol{\theta}_n-\boldsymbol{\mu}_{n,t})^{\sf T}\boldsymbol{\Sigma}_{n,t}(\boldsymbol{\theta}_n-\boldsymbol{\mu}_{n,t})-||\boldsymbol{\theta}_n||^2\nonumber\\&\hspace{1.2cm}\leq2[\ln(\text{det}(\boldsymbol{\Sigma}_{n,t}))-\ln(\sigma^D\delta)]\big\}\label{eq: Bayesian confidence sequence}
\end{align}
for all players $n$, where $\delta\in(0,1]$ is a hyperparameter. The confidence set \eqref{eq: Bayesian confidence sequence} indicates that with high probability, i.e., $\delta$ is sufficiently small, the mean estimate $\boldsymbol{\mu}_{n,t}$ is included.

\begin{lemma}[Utility Parameters Coverage Guarantee of PPR-UCB]\label{lemma: Bayesian confidence}
    Under Assumption \ref{assumption: bayesian linear regression}, the confidence set \eqref{eq: Bayesian confidence sequence} is anytime valid at level $1-\delta$ in the sense that it includes the ground truth parameters $\boldsymbol{\theta}_n^*$ with probability no smaller than $1-\delta$ for all time $t\geq 1$:
    \begin{align}
        \Pr(\boldsymbol{\theta}_n^*\in\mathcal{C}_{n,t},\,\,\text{for all $t\geq1$})\geq 1-\delta.\label{eq: Bayesian theta coverage}
    \end{align}
    In \eqref{eq: Bayesian theta coverage}, the probability is evaluated with respect to the ground truth distribution
    \begin{align}
        p(\mathcal{D}_{n,t}|\boldsymbol{\theta}_n^*)=\prod_{t'=1}^t p(y_{n,t'}|\mathbf{x}_{t'},\boldsymbol{\theta}_n^*).\label{eq: Bayesian lemma ground truth}
    \end{align}
\end{lemma}
\begin{proof}
    By taking the expectation of the prior-posterior ratio $\ell_{n,t}(\boldsymbol{\theta}^*_n)$ conditioned on the collected dataset $\mathcal{D}_{n,t-1}$ and applying Fubini's theorem, we have
    \begin{align}
        &\mathds{E}_{ p(y|\mathbf{x}_{t},\boldsymbol{\theta}^*_n)}[\ell_{n,t}(\boldsymbol{\theta}^*_n)|\mathcal{D}_{n,t-1}]\nonumber\\&=\mathds{E}_{ p(y|\mathbf{x}_{t},\boldsymbol{\theta}^*_n)}\bigg[\prod_{t'=1}^t\frac{\mathds{E}_{ p(\boldsymbol{\theta}_n|\mathcal{D}_{n,t})}p(y_{n,t'}|\mathbf{x}_{t'},\boldsymbol{\theta}_n)}{p(y_{n,t'}|\mathbf{x}_{t'},\boldsymbol{\theta}_n^*)}\bigg|\mathcal{D}_{n,t-1}\bigg]\\&=\mathds{E}_{ p(\boldsymbol{\theta}_n|\mathcal{D}_{n,t})}\Bigg[\ell_{n,t-1}(\boldsymbol{\theta}_n^*)\mathds{E}_{ p(y|\mathbf{x}_{t},\boldsymbol{\theta}^*_n)}\bigg[\frac{p(y_{n,t}|\mathbf{x}_t,\boldsymbol{\theta}_n)}{p(y_{n,t}|\mathbf{x}_t,\boldsymbol{\theta}_n^*)}\nonumber\\&\hspace{2.4cm}\bigg|\mathcal{D}_{n,t-1}\bigg]\Bigg]\\&\leq \ell_{n,t-1}(\boldsymbol{\theta}^*_n).\label{eq: bayesian supermartingale}
    \end{align}
    The inequality \eqref{eq: bayesian supermartingale} shows that the prior-posterior ratio sequence $\{\ell_{n,t}(\boldsymbol{\theta}_n^*)\}_{t\geq1}$ is a test \emph{supermartingale} for the ground truth distribution $p(\mathcal{D}_{n,t}|\boldsymbol{\theta}^*_n)$. By applying the Ville's inequality \cite{ville1939etude}, we have
    \begin{align}
        \Pr(\boldsymbol{\theta}_n^*\in\mathcal{C}_{n,t},\,\,\text{for all }t\geq1)&=1-\Pr(\exists\boldsymbol{\theta}_n^*,\ell_{n,t}(\boldsymbol{\theta}_n^*)>1/\delta)\\&\geq 1-\delta\mathds{E}[\ell_{n,t}(\boldsymbol{\theta}_n^*)]\\&=1-\delta,
    \end{align}
    which completes the proof.
\end{proof}

Using the confidence set $\mathcal{C}_{n,t}$ in \eqref{eq: Bayesian confidence sequence}, PPR-UCB then evaluates the confidence sets for the utility functions $u_n(\mathbf{x})$ and regret $f_n(\mathbf{x})$ for all BSs. 

\addtolength{\topmargin}{0.01in}
\begin{algorithm}[t]
\caption{PPR-UCB}\label{table: ppr-ucb}
\small                 % or \footnotesize if needed
\setlength{\algomargin}{0pt}   % remove left padding
\SetInd{0.3em}{0.3em}          % reduce block indents
\SetNlSkip{0.25em}             % tighten gap between line number and code
\DontPrintSemicolon            % saves a bit of horizontal space
\sloppy                        % avoid overfull boxes (prevents spill into gutter)

\SetKwInOut{Input}{Input}
\Input{Total number of iterations $T$, lengthscale $l$, parameter $\delta$}
\SetKwInOut{Output}{Output}
\Output{Optimized solution $\mathbf{x}^*$}\
Initialize observation dataset $\mathcal{D}_{n,0}=\emptyset$ for all BSs, iteration $t=0$\\
\While{\emph{$t\leq T$}}{
Evaluate the prior-posterior ratio $\ell_{n,t}(\boldsymbol{\theta}_n)$ as in \eqref{eq: pp ratio} for all BSs\\
Build the confidence set $\mathcal{C}_{n,t}$ as in \eqref{eq: Bayesian confidence sequence}\\
Obtain the utility confidence intervals $\{\mathcal{U}_{n,t}(\mathbf{x})\}_{n\in\mathcal{N}}$ as in \eqref{eq: Bayesian utility confidence}\\
Obtain the regret confidence intervals $\{\mathcal{F}_{n,t}(\mathbf{x})\}_{n\in\mathcal{N}}$ as in \eqref{eq: Bayesian regret ci} \\
Report the action profile $\tilde{\mathbf{x}}_{t+1}$ via \eqref{eq: reported x}\\
Select the exploring action profile $\mathbf{x}^{(n_{t+1})}$ via \eqref{eq: exploring strategy profile}\\
Produce the final decision $\mathbf{x}_{t+1}$ via \eqref{eq: x t+1}\\
Evaluate  $\mathbf{x}_{t+1}$ and update datasets as $\mathcal{D}_{n,t+1}=\mathcal{D}_{n,t}\cup(\mathbf{x}_{t+1},y_{n,t+1})$ for all BSs $n\in\mathcal{N}$\\
Update the GP posterior using $\mathcal{D}_{n,t+1}$ as in \eqref{eq: utility gp posterior mean} and \eqref{eq: utility gp posterior variance} for all BSs $n\in\mathcal{N}$\\
Set iteration index $t=t+1$
}
Return $\mathbf{x}^{*}=\mathbf{x}_{T}$\\
\end{algorithm}

\subsection{Acquisition Process}\label{ssec: acq process}
For the action profile acquisition process, PPR-UCB starts by evaluating a confidence interval for the utility function $u_n(\mathbf{x})$ of each BS $n$ at time $t$, which is obtained as
\begin{align}
    \mathcal{U}_{n,t}(\mathbf{x})&=[\check{u}_{n,t}(\mathbf{x}),\hat{u}_{n,t}(\mathbf{x})]\nonumber\\&=[\min\limits_{\boldsymbol{\theta}_n\in\mathcal{C}_{n,t}}\psi_n(\mathbf{x})^{\sf T}\boldsymbol{\theta}_n,\max\limits_{\boldsymbol{\theta}_n\in\mathcal{C}_{n,t}}\psi_n(\mathbf{x})^{\sf T}\boldsymbol{\theta}_n].\label{eq: Bayesian utility confidence}
\end{align}
Notably, unlike conventional confidence intervals-guided BO schemes \cite{srinivas2012information,chowdhury2017kernelized}, PPR-UCB does not require perfect knowledge on the utility function space, e.g., reproducing kernel Hilbert space (RKHS) norm. Using the confidence interval $\mathcal{U}_{n,t}(\mathbf{x})$ in \eqref{eq: Bayesian utility confidence} for the utility function, PPR-UCB also evaluates the confidence set for the regret $f_n(\mathbf{x})$ in \eqref{eq: player n regret} of each BS $n$ as
\begin{align}
    \mathcal{F}_{n,t}(\mathbf{x})&=[\check{f}_{n,t}(\mathbf{x}),\hat{f}_{n,t}(\mathbf{x})]\nonumber\\&=\bigg[\max\limits_{\mathbf{x}_n'}\check{u}_{n,t}(\mathbf{x}_n',\mathbf{x}_{-n})-\hat{u}_{n,t}(\mathbf{x}),\nonumber\\&\hspace{0.7cm} \max\limits_{\mathbf{x}_n'}\hat{u}_{n,t}(\mathbf{x}_n',\mathbf{x}_{-n})-\check{u}_{n,t}(\mathbf{x})\bigg].\label{eq: Bayesian regret ci}
\end{align}

Accordingly, the optimizer reports an action profile $\tilde{\mathbf{x}}_{t+1}$ for possible execution at next time $t+1$ via
\begin{align}
    \tilde{\mathbf{x}}_{t+1}=\arg\min\limits_{\mathbf{x}}\max\limits_{n\in\mathcal{N}}\check{f}_{n,t}(\mathbf{x}).\label{eq: reported x}
\end{align}
The intuition behind the reported action profile \eqref{eq: reported x} is to minimize the optimistic estimate $\check{f}_{n,t}(\mathbf{x})$ of the regret \eqref{eq: player n regret} obtained by any BS $n\in\mathcal{N}$ that has the strongest incentive to deviate. To allow for exploration, the central optimizer also searches for the BS $n_{t+1}$ with the maximum regret upper bound under the reported action profile $\tilde{\mathbf{x}}_{t+1}$ in \eqref{eq: reported x}, i.e.,
\begin{align}
    n_{t+1}=\arg\max\limits_{n\in\mathcal{N}}\hat{f}_{n,t}(\tilde{\mathbf{x}}_{t+1}),\label{eq: worst player}
\end{align}
and obtains an exploring action profile as
\begin{align}
    \mathbf{x}^{(n_{t+1})}=\Big[\tilde{\mathbf{x}}_{-n_{t+1},t+1}, \arg\max\limits_{\mathbf{x}'_{n_{t+1}}}\hat{u}_{n_{t+1},t}(\mathbf{x}'_{n_{t+1}},\tilde{\mathbf{x}}_{-n_{t+1},t+1})\Big]\label{eq: exploring strategy profile}
\end{align}
The exploring action profile $\mathbf{x}^{(n_{t+1})}$ attempts to improve the utility gain for the BS $n_{t+1}$ that the central optimizer believes to have the strongest incentive to deviate from $\tilde{\mathbf{x}}_{t+1}$.

Then, the final decision for the next action profile at time $t+1$ is selected between the reported action profile $\tilde{\mathbf{x}}_{t+1}$ in \eqref{eq: reported x} and the exploring action profile $\mathbf{x}^{(n_{t+1})}$ in \eqref{eq: exploring strategy profile} that brings higher uncertainty, i.e.,
\begin{align}
    \mathbf{x}_{t+1}=\arg\max\limits_{\mathbf{x}\in\{\tilde{\mathbf{x}}_{t+1},\mathbf{x}^{(n_{t+1})}\}}\sigma^2_{t}(\mathbf{x}).\label{eq: x t+1}
\end{align}
The rationale of the choice \eqref{eq: x t+1} refers to a double application of the so-called \emph{optimism in the face of uncertainty} (OFU) principle in \cite{lai1985asymptotically}, that is, to encourage maximally reducing uncertainty of potentially promising actions provided by \eqref{eq: reported x} and \eqref{eq: exploring strategy profile}. 

The overall procedure of PPR-UCB is summarized in Algorithm \ref{table: ppr-ucb}.

\begin{figure}[t]

%\begin{minipage}[b]{1.0\linewidth}
  \centering
  \centerline{\includegraphics[scale=0.34]{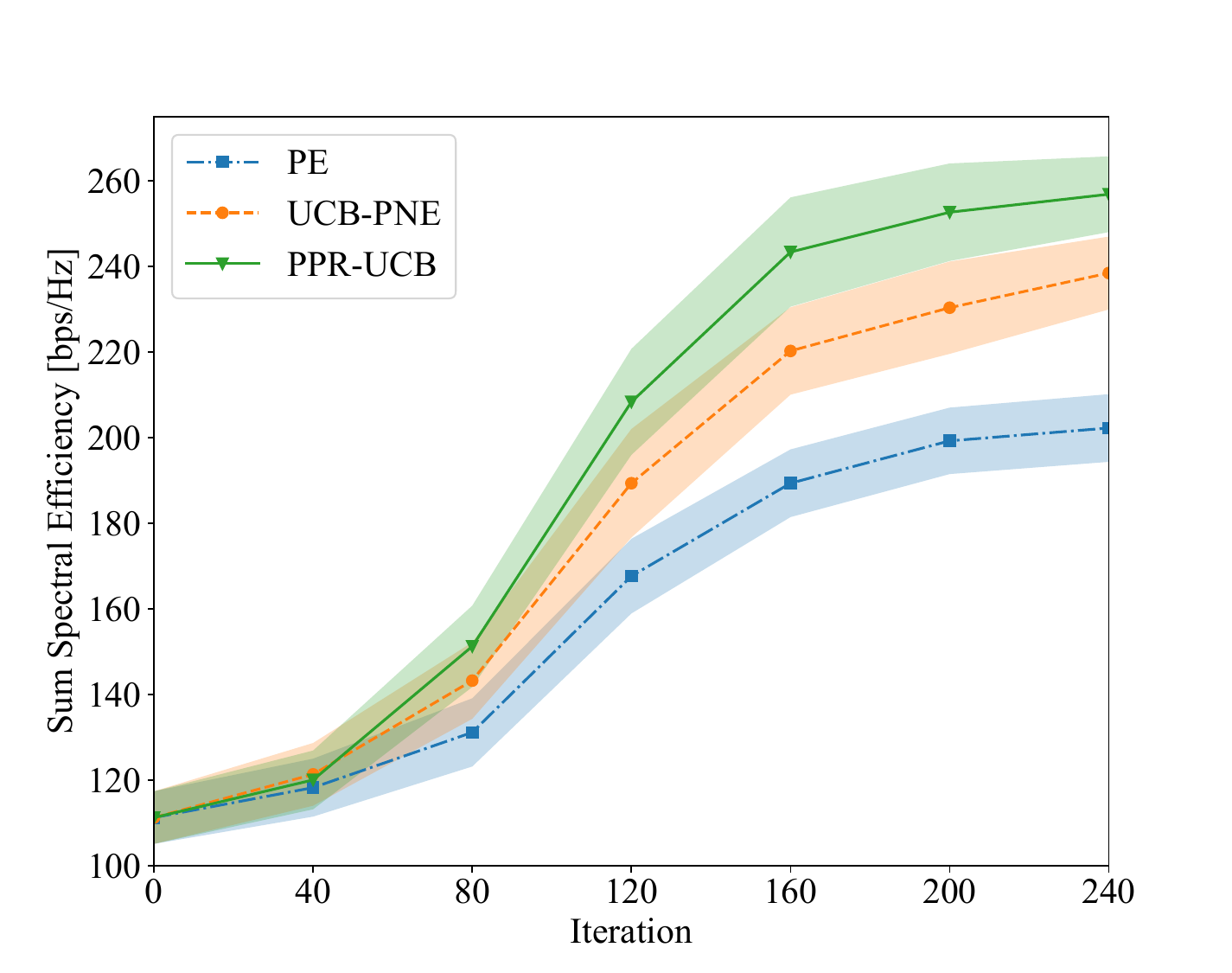}}
  \vspace{-0.2cm}
  \caption{Sum spectral efficiency against the number of optimization iterations $T$ for PE (blue dash-dotted line), UCB-PNE (orange dashed line), and PPR-UCB with parameter $\delta=0.05$ (green solid line).}
  \label{fig: sse vs iter}

\end{figure}

\section{Numerical Results}\label{sec: exp}
In this section, we empirically evaluate the performance of the proposed PPR-UCB on the non-cooperative game-based downlink power control problem introduced in Sec. \ref{sec: pf}. Throughout this section, the cellular communication network has $M=10$ UEs with $N_R=4$ receiving antennas in each BS with $N_T=16$ transmit antennas. The propagation pathloss $\text{PL}(d_{n,n'(m)})$ and channel fading models in \eqref{eq: channel matrix} are simulated following the urban microcellular (UMi) street canyon scenario as specified in 3GPP TR 38.901 \cite{3gpp38901}. 

For the cellular network topology configurations, we consider the service coverage of each BS to be a  planar grid with radius being $200$ meters, and UEs positions are uniformly generated in the interval $[20,200]$ meters. The maximum transmit power of each BS is constrained by $\sum_{m=1}^Mx_{n,m}\leq p_{\text{max}}=6.5$ Watt ($38.13$ dBm); the channel noise power is set to $\sigma_h^2=-86.46$ dBm; and the discount factor in \eqref{eq: utility} is set to $\lambda=0.1$.

\begin{figure}[t]

%\begin{minipage}[b]{1.0\linewidth}
  \centering
  \centerline{\includegraphics[scale=0.34]{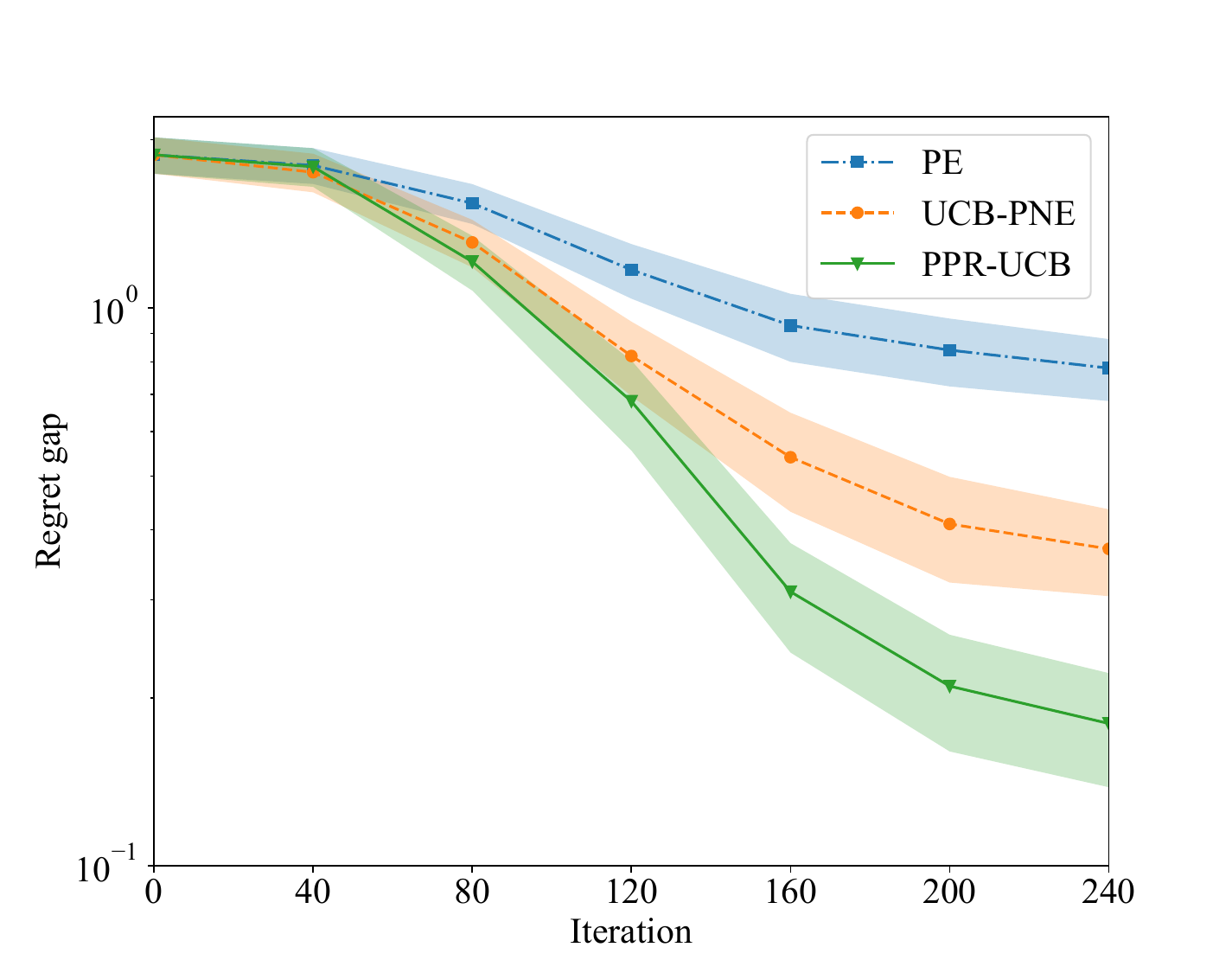}}
  \vspace{-0.2cm}
  \caption{Regret gap against the number of optimization iterations $T$ for PE (blue dash-dotted line), UCB-PNE (orange dashed line), and PPR-UCB with parameter $\delta=0.05$ (green solid line).}
  \label{fig: regret vs iter}

\end{figure}

We consider the following BO schemes as the benchmarks in all experiments are: 1) \emph{Probability of equilibrium} (PE) \cite{picheny2019bayesian}, which selects action profiles $\mathbf{x}_t$ with the goal of maximizing the probability of obtaining a PNE; and UCB-PNE in \cite{tay2023no}, which relies on Gaussian bandits optimization techniques \cite{chowdhury2017kernelized} to explore the $\epsilon$-PNE solutions. The GPs used in all schemes are configured with lengthscale $l=0.85$ and observation noise variance $\sigma^2=0.67$. The dissatisfaction tolerance level $\epsilon$ in \eqref{eq: epsilon pne} is selected as the best achievable level $\epsilon^*$ obtained by grid search, which is expressed as
\begin{align}
    \epsilon^*=\inf\{\epsilon\in\mathds{R}|\mathcal{X}^{(\epsilon)}\neq \emptyset\}.\label{eq: epsilon definition}
\end{align}
All results are averaged over 100 random realizations of observation noise signals and channel matrices in \eqref{eq: channel matrix}, with error bars to encompass 90\% confidence level. 

In Fig. \ref{fig: sse vs iter}, we set the parameter $\delta=0.05$ for PPR-UCB and plot the sum of discounted spectral efficiency across $N=7$ BSs as a function of the number of optimization iterations $T$. It is observed that PE obtains the worst performance, since the probability of equilibrium is zero for every action profile when the best achievable tolerance level $\epsilon^*>0$. Furthermore, our proposed PPR-UCB outperforms UCB-PNE after $40$ optimization rounds, and its improvement over all benchmarks is evident as the number of iterations increases, demonstrating the superior data efficiency of PPR-UCB in approximating the $\epsilon^*$-PNE solutions.

\begin{figure}[t]

%\begin{minipage}[b]{1.0\linewidth}
  \centering
  \centerline{\includegraphics[scale=0.34]{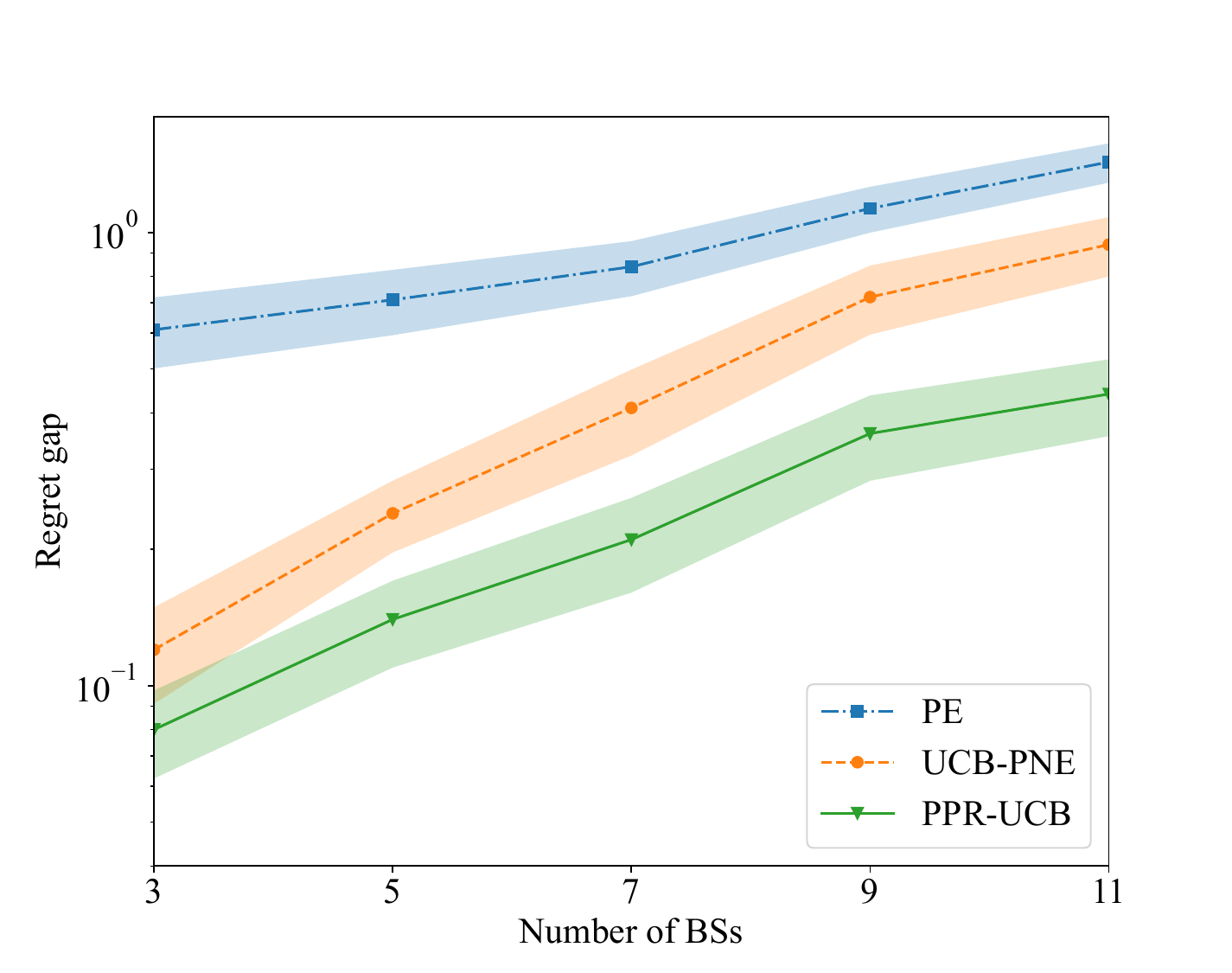}}
  \vspace{-0.2cm}
  \caption{Regret gap against the number of BSs $N$ for PE (blue dash-dotted line), UCB-PNE (orange dashed line), and PPR-UCB with parameter $\delta=0.05$ (green solid line).}
  \label{fig: regret vs n}
 
\end{figure}

In Fig. \ref{fig: regret vs iter}, we evaluate the \emph{regret gap} between $\epsilon^*$ with the best achievable $\epsilon$ obtained by each scheme, i.e., $\epsilon-\epsilon^*$, and plot the regret gap as a function of the number of optimization iterations $T$. Confirming the discussions in Fig. \ref{fig: sse vs iter}, the proposed PPR-UCB converges to its best achievable dissatisfaction level $\epsilon$ closer to the global optimum $\epsilon^*$ than all other benchmarks when the number of iterations $T>40$. These results demonstrate that the provident confidence sets evaluation approach adopted in PPR-UCB accurately captures the residual uncertainty of the utilities and regret for each BS, thus significantly improves the stability of the downlink transmission power allocation process.

Finally, we study the impact of the number of BSs $N$ on the regret gap in Fig. \ref{fig: regret vs n}. We set the total number of iterations to $T=200$ for all schemes. The regret gap of PPR-UCB is continuously better than other benchmarks on the increasing size of the cellular network, showing that the proposed PPR-UCB is scalable to strategic games with high dimensional action profiles. One way to further improve the performance of PPR-UCB on higher dimensional action spaces is to increase the size $D$ of the RFFs vector in \eqref{eq: Bayesian rff expression}, as this metric adjusts the accuracy of the kernel approximation in \eqref{eq: rff decomposition}.

\section{Conclusion}\label{sec: conclusion}
In this paper, we introduced a centralized BO framework for identifying approximate PNE solutions in a non-cooperative downlink power control game with costly-to-evaluate utility functions. By leveraging martingale techniques to construct confidence sets for the black-box utilities at each BS, the proposed PPR-UCB efficiently converges toward network configurations that approximate PNE solutions. Future work may address theoretical performance guarantees in the forms of game-theoretic regret bounds. Furthermore, it would be interesting to investigate the applications of BO to beam scheduling games \cite{xiong2024fair} or multi-fidelity optimizations in game settings \cite{zhang2025multi}.

\bibliographystyle{ieeetr}
\bibliography{refer}
\end{document}